\documentclass[12pt,a4,reqno,twoside]{amsart}
\usepackage{amssymb}
\usepackage{amsthm}
\usepackage{amsmath}
\usepackage{algorithm}
\usepackage{algpseudocode}
\usepackage{graphicx}

\newtheorem{theorem}{Theorem}[section]
\newtheorem{lemma}[theorem]{Lemma}

\newtheorem{corollary}[theorem]{Corollary}

\newtheorem{mydef}{Definition}

\newcommand{\oi}{\mathcal{O}}
\newcommand{\boi}{\partial \mathcal{O}}
\newcommand{\ideal}{\mathfrak{a}}

\providecommand{\OO}[1]{\operatorname{O}\bigl(#1\bigr)}

\title{Border basis detection is NP-complete}
\author[P.V. Ananth]{Prabhanjan V. Ananth}
\thanks{} 
\address{Dept. of Computer Science and Automation, Indian Institute of
Science}
\email{prabhanjan@csa.iisc.ernet.in}
\author[A. Dukkipati]{Ambedkar Dukkipati}
\thanks{} 
\address{Dept. of Computer Science and Automation, Indian Institute of
Science}
\email{ambedkar@csa.iisc.ernet.in}
\date{}

\begin{document}
\maketitle

\begin{abstract}
Border basis detection (BBD) is described as follows: given a set of generators of an ideal, decide whether that set of generators is a border basis of the ideal with respect to some order ideal. The motivation for this problem comes from a similar problem related to Gr\"obner bases termed as Gr\"obner basis detection (GBD) which was proposed by Gritzmann and Sturmfels (1993). GBD was shown to be NP-hard by Sturmfels and Wiegelmann (1996). In this paper, we investigate the computational complexity of BBD and show that it is NP-complete.  
\end{abstract}

\section{Introduction}
\noindent Gr\"obner bases play an important role in computational commutative algebra and algebraic geometry as they are used to solve classic problems like ideal membership, intersection and saturation of ideals, solving system of polynomial equations and so on. Gr\"obner bases are defined with respect to a `term order' and the choice of the term order plays a crucial role in time required to compute Gr\"obner bases. Gr\"obner bases are also known to be numerically unstable and hence are not suitable to be used to describe ideals which are constructed from measured data. Border bases, an alternative to Gr\"obner bases, is known to show more numerical stability as compared to Gr\"obner bases. 
\par The theory of border bases was used by Auzinger and Stetter~\cite{SolvingMultivarPoly:AuzStet} to solve zero dimensional polynomial systems of equations. There has been ongoing research to extend this solving technique for solving positive dimensional polynomial systems. The notion of border bases was introduced to find a system of generators for zero dimensional ideals having some nice properties. The theory was generalised to positive dimensional ideals by Chen and Meng~\cite{BBpositiveDim:Chen}. The connection between border bases and statistics was explored by Robbiano, Kruezer and Kehrein~\cite{AlgViewBB05}. Kehrein and Kreuzer gave characterizations of border bases~\cite{CharactBB} and also extended Mourrain's idea~\cite{Mourrain99:NFAlgo} to compute border bases~\cite{ComputingBB05}. The border bases as computed by the algorithm were associated with degree compatible term orderings. Brian and Pokutta~\cite{PolyApprBB2} gave a polyhedral characterisation of order ideals and gave an algorithm to compute border bases which was independent of term orderings. They also showed that computing a preference optimal order ideal is NP-hard.   
\par Gritzmann and Sturmfels~\cite{MinkowskiAdd:GritzSturmfels} introduced Gr\"obner basis detection (GBD) problem and solved this problem using Minkowski addition of polytopes. Later Sturmfels and Wiegelmann~\cite{GBD} showed that GBD is NP-hard. For this, they introduced a related problem called SGBD (Structural Gr\"obner basis detection) which was shown to be NP-complete by a reduction from the set packing problem. Using SGBD it was proved that GBD is NP-hard. We introduce a similar problem related to border bases known as Border Basis Detection (BBD) and prove that the problem is NP-complete. 
\par In \S~\ref{borderbases}, we give preliminaries for border bases and describe the border basis detection problem. In \S~\ref{mainresult}, we give a polynomial time reduction from 3,4-SAT to BBD and then we will show the correctness of the reduction. 

\section{Border Bases} 
\label{borderbases}    
\noindent Let $\mathbb{F}[x_1, \ldots ,x_n]$ be a polynomial ring, where $\mathbb{F}$ is a field. $\mathbb{T}^n$ denotes the set of terms \textit{i.e.}, $\mathbb{T}^n=\{{x_1}^{\alpha_1} \ldots {x_n}^{\alpha_n}: (\alpha_1, \ldots ,\alpha_n) \in \mathbb{Z}_{\geq 0}^{n}\}$. The total degree of a term $t={x_1}^{\alpha_1} \ldots {x_n}^{\alpha_n}$ denoted by $\mathrm{deg(t)}$ is $\sum_{i=1}^{n}{\alpha_i}$. We represent all the terms of total degree $i$ by $\mathbb{T}_{i}$ and all the terms of total degree less than or equal to $i$ by $\mathbb{T}_{\leq i}^{n}$.  By support of a polynomial we mean, all the terms appearing in that polynomial \textit{i.e,} support of a polynomial $f=\sum_{i=1}^{s}c_it_i$,  where $c_i \in \mathbb{F}$ and $t_i \in \mathbb{T}^n$ (denoted by $\mathrm{Supp}(f)$) is $\{t_1, \ldots ,t_s\}$. Similarly, support of a set of polynomials is the union of support of all the polynomials in the set \textit{i.e.}, $\mathrm{Supp}(S)= \displaystyle \bigcup_{f \in S} \mathrm{Supp}(f)$.\\
The following notions are useful for the theory of border basis.\\ 
\begin{mydef}
A non-empty finite set of terms $\oi$ $\subset$ $\mathbb{T}^n$ is called an order ideal if it is closed under forming divisors \textit{i.e.}, if $t \in \oi$ and $t'|t$ then it implies $t' \in \oi$. 
\end{mydef}

\begin{mydef}
Let $\oi$ be an order ideal. The border of $\oi$ is the set $$\partial \oi = (\mathbb{T}_{1}^n . \oi) \backslash \oi = (x_1 \oi \cup ... \cup x_n \oi) \backslash \oi.$$ The first border closure of $\oi$ is defined as the set $ \oi \cup \partial \oi$ and it is denoted by $\overline{ \partial \oi}$.
\end{mydef}
\noindent It can be shown that $\overline{ \partial \oi}$ is also an order ideal.
        
%Definition of O-border prebasis
\begin{mydef}
Let $\oi$ = $\{t_1,...,t_\mu\}$ be an order ideal, and let $\partial \oi$ = $\{ b_1,...,b_\nu \}$ be it's border. A set of polynomials $G$ = $\{g_1,...,g_\nu \}$ is called an $\oi$-border prebasis if the polynomials have the form $g_j$ = $b_j$ - $\sum\limits_{i=1}^{\mu} \alpha_{ij}t_i$, where $\alpha_{ij} \in \mathbb{F}$ for $1 \leq i \leq \mu$ and $1 \leq j \leq \nu$.
 \end{mydef} 

\noindent Note that the $\oi$-border prebasis consists of polynomials which have exactly one term from $\partial \oi$ and rest of the terms are in order ideal $\oi$. 
\par If a $\oi$-border prebasis belongs to an ideal $\ideal$ and the order ideal has a nice property with respect to an ideal then that $\oi$-border prebasis is termed as $\oi$-border basis. The definition of $\oi$-border basis is given below.

%Definition of O-border basis 
\begin{mydef}
Let $\oi = \{t_1,\ldots , t_\mu\}$ be an order ideal and $G = \{g_1, \ldots , g_\nu \}$ be an $\oi$-border prebasis consisting of polynomials in $\ideal$. We say that the set $G$ is an \textbf{$\oi$-border basis} of $\ideal$ if the residue classes of $t_1, \ldots , t_\mu$ form a $\mathbb{F}$-vector space basis of $\mathbb{F}[x_1,\ldots,x_n]/\ideal$.
\end{mydef}

\noindent It can be shown that an $\oi$-border basis of an ideal $\mathfrak{a}$ indeed generates $\mathfrak{a}$~\cite{AlgViewBB05}. It can also be shown that for a fixed order ideal $\oi$, with respect to an ideal $\mathfrak{a}$ there can be at most one $\oi$-border basis for $\mathfrak{a}$. In \cite{CharactBB}, a criterion was stated for an $\oi$-border prebasis to be $\oi$-border basis termed as ``Buchberger criterion for border bases". The following notion is required for stating that criterion.

\begin{mydef}
Let $G = \{ g_1, \ldots , g_\nu \}$ be an $\oi$-border prebasis. Two prebasis polynomials $g_k,g_l$ are neighbors, where $k,l \in \{1, \ldots ,\nu\}$, if their border terms are related according to $x_i b_k = x_jb_l$ or $x_ib_k=b_l$ for some indeterminates $x_i, x_j$. Then, the corresponding S-polynomials are
$$S(g_k,g_l)= x_ig_k - x_jg_l \mbox{ and } S(g_k,g_l)= x_ig_k - g_l$$
respectively.
\end{mydef}

\noindent We now state the Buchberger criterion for border bases.

\begin{theorem}
An $\oi$-border prebasis $G = \{g_1, \ldots , g_\nu\}$ is an $\oi$-border basis of an ideal $\ideal$ if and only if $G \subset \ideal$ and, for each pair of neighboring prebasis polynomials $g_k,g_l$, there are constant coefficients $c_j \in \mathbb{F}$ such that
$$S(g_k,g_l) = c_1g_1+ \ldots + c_\nu g_\nu.$$
\end{theorem}
\noindent The proof for the above theorem can be found in \cite{CharactBB}. In the next section, we state BBD and give our result.\\

%%%%%%%%%%%%%%%%%%%%%%%%%%%%%%%%%%%%%%%%%%%%%%%%%%%%%%%%%%%%%%%%%% Result %%%%%%%%%%%%%%%%%%%%%%%%%%%%%%%%%%%%%%%%%%%%%%%%%%%%%%%%%%%%%%%%%%%%%%%%%%%%%%%%%%%%%%%%%%%%

\section{Result}
\label{mainresult}
\noindent BBD is described as follows:

\begin{quote}
Given a set of polynomials $\mathcal{F}$ such that $\ideal=\langle \mathcal{F} \rangle$ where $\ideal$ is an ideal, decide whether $\mathcal{F}$ is a $\oi$-border basis of $\ideal$ for some order ideal $\oi$.
\end{quote}

\noindent We first describe the input representation of the polynomials for the BBD instance.  We follow the ``sparse representation" as in \cite{MinkowskiAdd:GritzSturmfels} to represent the polynomials in $\mathcal{F}$. Let $\mathbb{F}[x_1, \ldots ,x_n]$ be the polynomial ring under consideration and let $\mathcal{F}$ be the set of input polynomials in the BBD instance. Consider a polynomial $f=c_1X^{\alpha_1}+ \ldots +{c_s}X^{\alpha_s} \in \mathcal{F}$ where $c_i \in \mathbb{F},\ X^{\alpha_i}={x_1}^{\alpha_{1i}} \ldots {x_n}^{\alpha_{ni}}$ for $i \in \{1, \ldots ,s\}$ and $\alpha_i=(\alpha_{1i}, \ldots ,\alpha_{ni}) \in \mathbb{Z}_{\geq 0}^{n}$. $f$ is represented by it's nonzero field coefficients $c_1, \ldots ,c_k$ and it's corresponding nonnegative exponent vectors $\alpha_1, \ldots ,\alpha_s$. 
\par In this section, we show that BBD is NP-complete. The NP-complete problem we have chosen for our reduction is 3,4-SAT. 3,4-SAT denotes the class of instances of the satisfiability problem with exactly three variables per clause and each variable  or it's complement appears in no more than four clauses. The 3,4-SAT problem was shown to be NP-complete by Tovey~\cite{SimplifiedNPCsat:Tovey}.
\par  Let $\mathcal{I}$ be an instance for the 3,4-SAT problem. Let $X_1, \ldots ,X_n$ be variables and $C_1, \ldots ,C_m$ be clauses in $\mathcal{I}$ such that $\mathcal{I}=C_1 \wedge C_2 \ldots \wedge C_m$. Each clause is a disjunction of three literals. For example, $(X_i \vee \overline{X}_j \vee X_k)$ represents a clause for $i,j,k \in \{1, \ldots ,n\}$. Assume without loss of generality that $X_i$ appears in at least one clause and so does $\overline{X}_i$. Also assume that $X_i$ and $\overline{X}_i$ do not appear in the same clause for any $i \in \{1, \ldots ,n\}$. We construct a BBD instance from this 3,4-SAT instance.\\

%\noindent It can be seen that BBD is in NP: If we guess $\oi$, 

\noindent Consider the polynomial ring $$P=\mathbb{F}[x_1, \ldots ,x_n, \overline{x}_1, \ldots , \overline{x}_n, c_1, \ldots ,c_m, x_{c_1}, \ldots ,x_{c_m}, X],$$ where $\mathbb{F}$ is a field. We will reduce the 3,4-SAT instance $\mathcal{I}$ to a set of polynomials $\mathcal{F} \subset P$. Note that $P$ is a polynomial ring with $N=2n+2m+1$ indeterminates. Before we describe the reduction, we list some definitions and observations that will be useful for our reduction.

\begin{itemize}
\item With respect to all the clauses in which $X_i,\overline{X}_i$ appear for $i \in \{1, \ldots ,n\}$, we associate the term $t_{C_{x_i}}=\left( \displaystyle\prod_{j \in S} c_j \right) X^{\alpha}$ where for each $j \in S \subset \{1, \ldots ,m\}$ either $X_i$ or $\overline{X}_i$ appears in $C_j$ and $\alpha=4-|S|$. Note that $\mathrm{deg}(t_{C_{x_i}})=4$.

\item With respect to each $X_i,\overline{X}_i$ for $i \in \{1, \ldots ,n\}$, we associate the terms $t_{X_i}=x_i \overline{x}_i^2 t_{C_{x_i}}$, $t_{\overline{X}_i}={x_i}^2 \overline{x}_i t_{C_{x_i}}$ respectively. Note that $\mathrm{deg}(t_{X_i})=\mathrm{deg}(t_{\overline{X}_i})=7$.

\item We define children of a term $t$ to be $$k(t)=\{t'|\mbox{ for some indeterminate } y,\ t'y=t\}.$$ Note that each term can have at most $N$ children.

\item Extending the above definition, we define children of a set of terms $S$ to be $k(S)=\displaystyle\bigcup_{t \in S}k(t)$. It follows that for two sets of terms $A$ and $B$, $k(A \cup B)=k(A) \cup k(B)$.

\item We define parents of a term $t$ to be $$p(t)=\{t'|\mbox{ for some indeterminate } y,\ ty=t'\}.$$ Note that each term has exactly $N$ parents. 

\item Extending the above definition, we define parents of a set of terms $S$ to be $p(S)=\displaystyle\bigcup_{t \in S}p(t)$.

\item $K_{X_i}=\Big\{ {\frac{t_{X_i}x_{c_l}}{c_l}}\footnote{This notation is used for convinience. If $t',t$ are terms such that $t'x=t$ for some indeterminate x then we represent $t'$ as $\frac{t}{x}$.}  \Big|\ X_i \mbox{ appears in clause } C_l \mbox{ for some }l \in \{1, \ldots ,m\}  \Big\}$ for $i = 1, \ldots ,n$.

\item $K_{\overline{X}_i}=\Big\{ {\frac{t_{\overline{X}_i}x_{c_l}}{c_l}}\footnotemark[\value{footnote}] \Big|\ \overline{X}_i \mbox{ appears in clause } C_l \mbox{ for some }l \in \{1, \ldots ,m\} \Big\}$ for $i = 1, \ldots ,n$.

\item $K_i=K_{X_i} \cup K_{\overline{X}_i} \cup \{t_{X_i},t_{\overline{X}_i} \}$ for $i = 1, \ldots ,n$. 

\item $P_{X_i}=\Big\{t_{X_i}x_{c_l} \Big|\ X_i \mbox{ appears in clause } C_l \mbox{ for some }l \in \{1, \ldots ,m\} \Big\}$ for $i = 1, \ldots ,n$.

\item $P_{\overline{X}_i}=\Big\{t_{\overline{X}_i}x_{c_l} \Big|\ \overline{X}_i \mbox{ appears in clause } C_l \mbox{ for some }l \in \{1, \ldots ,m\} \Big\}$ for $i = 1, \ldots ,n$.

\item $P_i= P_{X_i} \cup P_{\overline{X}_i}$ for $i= 1, \ldots ,n$. The number of clauses $X_i$ or $\overline{X}_i$ appear is $|P_i|$. Hence, $|P_i| \leq 4$.

\item We define $I(t)$ to be the number of indeterminates that divide a term $t$. Note that $I(t)=k(t)$.

\item Region associated with $X_i,\overline{X_i}$ for $i \in \{1, \ldots ,n\}$ is defined as $$R_i=\ k(P_i)\ =\ k(P_{X_i}) \cup k(P_{\overline{X}_i}).$$ In other words $R_i$ consists of all the children of $P_i$ and hence $|R_i| \leq 4N$. For $i,j \in \{1, \ldots ,n\}$ and $i \neq j$, since every term in $R_i$ contains either $x_i$ or $\overline{x}_i$ (and does not contain $x_j,\overline{x}_j$) and similarly every term in $R_j$ contains either $x_j$ or $\overline{x}_j$ (and does not contain $x_i,\overline{x}_i$) and hence  $R_i \cap R_j=\phi$.

\item For a set $S \subset \mathbb{T}^n$, we define $\mathrm{maxdeg}(S)$ to be $\displaystyle\mathrm{max}_{t \in S}{\mathrm\{deg}(t)\}$ \textit{i.e.} $\mathrm{maxdeg}$ is a function from power set of $\mathbb{T}^n$ to $\mathbb{N}$ which associates to each set $S \subset \mathbb{T}^n$ a number $n$ such that there exists a term in $S$ having total degree $n$ and no term in $S$ has total degree greater than $n$.  
\end{itemize}

\noindent We now state and prove a few observations that will be used for the reduction.\\

\begin{lemma}
Two distinct terms can have no more than one common parent \textit{i.e.}, for two distinct terms $t_1,t_2$, $|p(t_1) \cap p(t_2)| \leq 1$.
\end{lemma}
\begin{proof}
Consider two terms $t_1,t_2$ such that $t_1 \neq t_2$. Assume that there exists two distinct terms $t,t'$ such that $t_1,t_2 \in k(t)$ and $t_1,t_2 \in k(t')$. This implies that there exists indeterminates $y_1,y_2,{y}'_1,{y}'_2$ such that\\
$$t_1y_1=t,\ t_2y_2=t,\ t_1{y}'_1=t',\ t_2{y}'_2=t'.$$
This implies that ${y}'_2y_1={y}'_1y_2$. Since, $y_1 \neq {y_1}'$ and $y_1 \neq y_2$, we get a contradiction.
\end{proof} 

\begin{corollary}
For two distinct terms $t_1,t_2$, $|k(t_1) \cap k(t_2)| \leq 1$.
\end{corollary}
\begin{proof}
This follows from the definition and the previous lemma.
\end{proof}

\begin{corollary}
\label{boundsamep}
Let $S$ be a set of terms and $t$ be a term such that $t \notin S$. Then $|k(t) \cap k(S)| \leq |S|$. 
\end{corollary}
\begin{proof}
Let $S=\displaystyle\bigcup_{i:a_i \in S}\{a_i\}$. We have $$k(t) \cap k(S)=\displaystyle\bigcup_{i:a_i \in S}(k(t) \cap k(a_i)).$$ But,
\begin{eqnarray*}
\Big| \displaystyle\bigcup_{i:a_i \in S}(k(t) \cap k(a_i)) \Big| &\leq& \sum_{i:a_i \in S}|(k(t) \cap k(a_i))|\\ 
&\leq& |S| \mbox{ (from the previous corollary)}.
\end{eqnarray*}
Hence, $|k(t) \cap k(S)| \leq |S|$. 
\end{proof}

\begin{lemma}
\label{regsamep}
No two terms from two different regions can have a common parent \textit{i.e.}, if there are two terms $t_1 \in R_i,\ t_2 \in R_j$ then there exists no term $t_3$ such that $t_1,t_2 \in k(t_3)$. 
\end{lemma}
\begin{proof}
Let $t_1 \in R_i$ and $t_2 \in R_j$ for some $i,j \in \{1, \ldots ,n\}$. Assume without loss of generality that $t_1 \in k(t_{X_i}y)$ (a similar argument holds if $t_1 \in k(t_{\overline{X}_i}y)$), where $y$ is an indeterminate such that $t_{X_i}y \in P_{X_i}$. Hence, there exists an indeterminate $y'$ such that $t_1y'=t_{X_i}y$. Now, if we assume that there exists a term $t_3$ such that $t_1,t_2 \in k(t_3)$ then there exists two indeterminates $y_1,y_2$ such that,
$$t_3=t_1y_1=t_2y_2 \Rightarrow t_{X_i}yy_1=t_2y_2y'.$$
But, $x_i{\overline{x}_i}^2|t_{X_i} \Rightarrow x_i{\overline{x}_i}^2|t_2y_2y' \Rightarrow x_i{\overline{x}_i}^2|y_2y'$ (since $x_i,\overline{x}_i$ does not divide any term in $R_j$) and hence a contradiction. 
\end{proof}

\begin{lemma}
\label{notallinbo}
Let $\oi$ be an order ideal. If all the children of a term $t$ are in $\boi$ then $t$ cannot be in $\boi$ and $\oi$ \textit{i.e.}, for a term $t$ such that $k(t) \subset \boi$ then $t \notin \oi, t \notin \boi.$
\end{lemma}
\begin{proof}
Let $t$ be a term such that $k(t) \subset \boi$. If $t \in \oi$ then $k(t) \subset \oi$ and hence $t \notin \oi$. If $t \in \boi$ then there exists some indeterminate $y'$ such that for some term $t' \in \oi$, we have $t'y'=t$. But $t' \in k(t) \Rightarrow t' \in \boi$, a contradiction. Hence, $t \notin \boi$.
\end{proof}

\begin{lemma}
\label{boundind}
For a term $t$ such that $t \in k(P_i)$ where $i \in \{1, \ldots ,n\}$, then $I(t) \geq |P_i|+2$.
\end{lemma}
\begin{proof}
For a term $t' \in P_i$, $I(t')=3+I(t_{C_{x_i}})$, but
\begin{eqnarray*}
I(t_{C_{x_i}})& = &\mathrm{min}(\mbox{number of clauses in which }X_i,\overline{X}_i \mbox{ appear}+1,4)\\
& = &\mathrm{min}(|P_i|+1,4).
\end{eqnarray*}
We have $I(t')=\mathrm{min}(|P_i|+1,4)+3$ and thus for $t \in k(t')$, 
$$I(t) \geq \mathrm{min}(|P_i|+1,4)+2\ =\ \mathrm{min}(|P_i|+3,6) \mbox{ and since $|P_i| \leq 4$,}$$
$$I(t) \geq |P_i|+2.$$ 
\end{proof}

\begin{lemma}
\label{fact1}
Let $t_1,t_2$ be terms such that $t_1t=t_2$ where $t$ is a term and $t \neq 1$. If $x$ is an indeterminate such that $x$ divides $t$ then $t_1\Big|\frac{t_2}{x}$.
\end{lemma}
\begin{proof}
Since $x$ divides $t$, $x$ also divides $t_2$ and hence $\frac{t_2}{x},\frac{t}{x}$ are valid terms. We have, $t_1\Big(\frac{t}{x}\Big)=\frac{t_2}{x}$. Thus, $t_1\Big|\frac{t_2}{x}$. 
\end{proof}
In other words, the above lemma states that if a term $t_1$  divides $t_2$ and $t_1 \neq t_2$, then there exists a child of $t_2$, say $t_3$ such that $t_1$ divides $t_3$.

%%%%%%%%%%%%%%%%%%%%%%%%%%%%%% BBD is in NP %%%%%%%%%%%%%%%%%%%%%%%%%%%%

\subsection{BBD is in NP}
We ask the following question: When does a set of terms be a border with respect to an order ideal. It turns out that if the terms in $B$ obey some conditions then there exists an order ideal such that $B$ is it's border. 
\par Let $B \subset \mathbb{T}^n$ be a finite set of terms. Let $B'$ be a subset of $B$ such that every term $t$ in $B'$ obeys the following conditions:\\
1) For indeterminates $y,x$ such that $x|t$ and $y \neq x$, atleast one of $ty,\frac{ty}{x},\frac{t}{x}$ is in $B$.\\
2) There exists an indeterminate $x$ such that $x|t$ and $\frac{t}{x} \notin B$.\\ 
3) Let $t',t''$ be terms such that $t'|t'',t''|t$ and $t''$ is the parent of $t'$. If $t' \in B$ then $t''$ is in $B$.\\  
If $B=B'$ then we say that ``$B$ satisfies the three conditions" else we say that ``$B$ does not satisfy the three conditions". We will later prove that the three conditions mentioned before are sufficient and necessary for the existence of an order ideal such that $B$ is it's border. Before that we state an equivalent formulation of third condition.\\
For a term $t$ in $B$ consider the following set:\\
$$S_{t}=\Big\{t'' \in \mathbb{T}^n \Big|\ t''|t \mbox{ and } \exists \mbox{ a term } t' \in B \mbox{ such that } t'|t''\}$$
\begin{lemma}
\label{thirdcond}
All the terms in $B$ obey the third condition if and only if $S_{t} \subset B$ for all $t \in B$.
\end{lemma}
\begin{proof}
If for all $t \in B$, $S_{t} \subset B$ then $B$ satisfies the third condition.
\par Assume all terms in $B$ obey the third condition. Let $t$ be a term in $B$ and let ${S'_{t}}$ be the subset of $S_{t}$ such that it contains all the terms in $S_{t}$ and not in $B$. If $S'_t=\emptyset$ then $S_t \subset B$. Hence assume that $S'_t \neq \emptyset$. Let $t''$ be a term in $S'_{t}$ such that no term in $S'_t$ divides $t''$. Since $t'' \in S_t$, there exists a term $t_1$ such that $t_1|t''$ and $t_1 \in B$. From lemma \ref{fact1}, $t_1|t'$ where $t' \in k(t'')$. Since $t_1|t',t'|t$ and $t_1 \in B$, we have $t' \in S_t$. By the choice of $t''$, $t' \in S'_t$ which means $t' \in B$. We have a situation where there are three terms $t,t',t''$ such that \textit{(i)} $t'|t'',t''|t$, \textit{(ii)} $t,t' \in B$, $t'' \notin B$ and \textit{(iii)} $t'' \in p(t')$. But this contradicts the fact that all the terms in $B$ satisfy the third condition.  
\end{proof}  
\noindent From the above lemma, for a term $t \in B$ the third condition can be rephrased as follows: \\
3a) For terms $t',t''$ such that $t' \in B,\ t'|t''$ and $t''|t$ then $t''$ is in $B$. \\

\noindent We now give the necessary and sufficient conditions for $B$ to be the border of an order ideal $\oi$. 

\begin{theorem}
There exists an order ideal $\oi$ such that $\boi=B$ if and only if $B$ satsifies all the conditions.
\end{theorem}
\begin{proof}
Let $\oi$ be an order ideal such that $B$ is it's border \textit{i.e.} $B=\boi$. Assume that $B$ does not satisfy the three conditions which means there exists a term $t \in B$ which does not obey all the three conditions. Consider the following cases:\\
\textit{Case (i)} Suppose $t$ does not obey the first condition. There exists indeterminates $x,y$ such that $x|t,y \neq x$ and $t_1=ty \notin B,t_2=\frac{ty}{x} \notin B,t_3=\frac{t}{x} \notin B$. Since $t \in \boi$, $t_3$ is in $\oi$ which implies that $t_3y=t_2 \in \oi$ since $t_2 \notin \boi$. Similarly, $t_2x=t_1 \in \oi$. But $\oi$ is an order ideal and since $t|t_1$, $t$ should be in $\oi$ and hence a contradiction.\\ 
\textit{Case (ii)} Suppose $t$ does not obey the second condition. Then $k(t) \subset B=\boi$. From lemma \ref{notallinbo}, $t \notin \boi$ which is a contradiction.\\
\textit{Case (iii)} Suppose $t$ does not obey the third condition. There exists two terms $t',t''$ such that $t' \in B,t'' \in \oi$ and $t'|t'',t''|t,t'' \in p(t')$. Since $\oi$  is an order ideal, $t'' \in \oi$ implies that $t' \in \oi$, a contradiction.\\
Hence $B$ has to satisfy the three conditions for it to be the border of the order ideal $\oi$.
\par Assume that $B$ satisfies all the three conditions. Now, consider the following set:
$$\oi=\Big\{t \in \mathbb{T}^n \Big| \mbox{ there exists a term $t' \in B$ such that $t|t'$ and $t \notin B$ }\Big\}$$

\noindent \textbf{Claim. }$\oi$ is an order ideal.\\
\textit{Proof. }Consider a term $t \in \oi$. Let $t'$ be a term such that $t'|t$. By the construction of $\oi$, there exists a term $t'' \in B$ such that $t|t''$ and this implies that $t'|t''$. Now, if 
$t'$ was in $B$ then from lemma \ref{thirdcond}, $t''$ would violate the third condition and hence $t' \notin B$. Hence, $t' \in \oi$.\\

\noindent \textbf{Claim. }$B=\boi$.\\
\textit{Proof. }We will first show that $B \subset \boi$. Consider a term $t \in B$ and from the second condition there exists a term $t' \notin B$ such that $t=t'x$ for some indeterminate $x$. This implies that $t' \in \oi$ and hence, $t'x=t \in \boi$ since $t \notin \oi$. It remains to show that $\boi \subset B$. Let $t_1 \in \boi$ and hence there exists a term $t \in \oi$ such that $tx=t_1 \in \boi$ for an indeterminate $x$. From the construction of $\oi$, $t$ divides atleast one term in $B$. Let $t_2 \in B$ such that $t|t_2$ and if there is a term $t'$ such that $t|t'$ and $t'|t_2$ then $t' \in \oi$. Since $t|t_2$, from lemma \ref{fact1} there exists a child of $t_2$ such that $t$ divides that term. Let $x_1$ be an indeterminate such that $x_1|t_2$ and $t \Big| \frac{t_2}{x_1}$. Consider the following two cases:\\
\textit{Case (i)} $x_1=x$: In this case $t_1|t_2$ and hence $t_1 \in B$ since $t_1 \notin \oi$.\\
\textit{Case (ii)} $x_1 \neq x$: From the first condition, one of $t_2x,\frac{t_2x}{x_1},\frac{t_2}{x_1}$ has to be in $B$. Assume that $ \frac{t_2}{x_1} \in B$. We have a term ${t_2}''=\frac{t_2}{x_1}$ such that $t|{t_2}'',{t_2}''|{t_2}$ and ${t_2}'' \in B$ which contradicts the choice of $t_2$. Hence $\frac{t_2}{x_1} \notin B$ which means $\frac{t_2x}{x_1}$ or $t_2x$ is in $B$. Now $t \Big| \Big( \frac{t_2}{x_1} \Big)$ and hence $tx \Big| \Big( \frac{t_2x}{x_1} \Big)$, $tx \Big| t_2x$ which implies that $tx=t_1$ divides a term in $B$. This further implies that $t_1 \in \oi$ or $t_1 \in B$. Since $t_1 \in \boi, t_1 \notin \oi$ and thus $t_1 \in B$.
\end{proof}

\noindent Let $B$ be a set of terms and let $m$ be the size of binary representation of $B$.\\
For a term $t \in B$ and a fixed pair of indeterminates $(y,x)$, we can search whether $\frac{ty}{x},\frac{t}{x},ty$ are in $B$ in $\OO{m}\footnote{Big-O notation}$ time. And since there are $|B| (\leq m)$ terms and $N^2$ pairs of indeterminates ($N$ is the number of indeterminates), condition 1 can be checked in $\OO{m^2N^2}$ time.\\
For every term $t$, atmost $N$ children are possible. In $\OO{Nm}$ time it can be checked whether all the children of the term $t$ are in $B$ or not. Since there are $|B|$ terms, condition 2 can be checked in $\OO{Nm^2}$ time.\\
Every term has exactly $N$ parents. For terms $t',t' \in B$ fixed such that $t'|t$, it takes $\OO{Nm}$ time to check whether all the parents of $t'$ dividing $t$ are in $B$. Since there are $|B|^2$ such terms possible, condition 3 can be checked in $\OO{Nm^3}$ time.\\
Hence, it can be checked in time polynomial in $N$ and $m$ (binary size of $B$) whether $B$ is the border of some order ideal.\\ 

\noindent Let $B$ be the border of some order ideal $\oi$ \textit{i.e.} $B=\boi$ and let $\mathcal{F}$ be a set of polynomials such that the support of each polynomial in $\mathcal{F}$ contains exactly one term from $B$ and $|B|=|\mathcal{F}|$. We state a lemma that will be helpful in checking whether $\mathcal{F}$ is a $\oi$-border prebasis.

\begin{lemma}
 $\mathcal{F}$ is a $\oi$-border prebasis if and only if every term in $\mathrm{Supp}(\mathcal{F}) \backslash B$ divides a term in $B$.
\end{lemma}
\begin{proof}
Let $\mathcal{F}$ be a $\oi$-border prebasis. Then $B'=\mathrm{Supp}(\mathcal{F} \backslash B) \subset \oi$. Let $t \in B'$ \textit{i.e.} $t \in \oi$. For an indeterminate $x$, consider the sequence of terms $t,tx,tx^2, \ldots$. Not all the terms in the sequence can be in $\oi$ since $\oi$ is a finite set of terms. Let $i$ be the least number such that $tx^i \notin \oi$ and hence $tx^i \in \boi$. Thus, $t$ divides a term in $\boi$.
\par Let $t$ be a term in $B'$ such that $t$ divides a term $t' \in B$. As mentioned before, $\overline{\boi}$ is an order ideal and hence $t \in \overline{\boi}$. Since, $t \notin \boi$, $t$ has to be in $\oi$. Thus, $B' \subset \oi$. Hence, $|B|=|\mathcal{F}|$ and support of each polynomial in $\mathcal{F}$ contains exactly one term in $B$ and the rest of the terms are in $\oi$. Thus, $\mathcal{F}$ is a $\oi$-border prebasis.
\end{proof}

We now give the proof that BBD is in NP.
\begin{theorem}
BBD is in NP.
\end{theorem}
\begin{proof}
Let $\mathcal{F}$ be a set of input polynomials to the BBD instance such that $\ideal= \langle \mathcal{F} \rangle$. Assume that a set $B=\mathrm{Supp}(\mathcal{F})$ containing exactly one term from each polynomial in $\mathcal{F}$ and $|B|=|\mathcal{F}|$, is given as a ``YES" certificate\footnote{A ``YES" certificate is a proof to show that $\mathcal{F}$ corresponds to an "yes" answer in BBD \textit{i.e.} $\mathcal{F}$ is a border basis of $\ideal$ with respect to some order ideal.} for $\mathcal{F}$ such that $B=\boi$ for some order ideal $\oi$ and $\mathcal{F}$ is a $\oi$-border basis. Let the binary size of representation of $\mathcal{F},B$ be denoted by $m_{\mathcal{F}},m_{B}$ respectively. This certificate can be verified in polynomial time as follows:\\
We have seen that it can be verified in time polynomial in $m_{B}$ and $N$ whether $B$ is the border of some order ideal $\oi$. In order to check whether $\mathcal{F}$ is a $\oi$-border prebasis, from the previous claim we need to check whether each term in $\mathrm{Supp}(\mathcal{F}) \backslash B$ divides a term in $B$. This can be implemented in $\OO{m_{\mathcal{F}}m_{B}}$ time. And in time polynomial in $m_{\mathcal{F}}$, it can be verified whether $\mathcal{F}$ satisfies the Buchberger criterion. Since a ``YES" certificate for the BBD instance can be verified in polynomial time, BBD is in NP.
\end{proof}    
 
\noindent We now give a polynomial time reduction from 3,4-SAT to BBD.
%%%%%%%%% Reduction

\subsection{Reduction}
We are now going to construct a set of polynomials $\mathcal{F}$ as follows:\\

\begin{itemize}
\item With respect to variable $X_i$ for $i \in \{1, \ldots ,n\}$, associate a polynomial $$t_{X_i}+t_{\overline{X}_i}.$$ 
We shall refer to such polynomials as \textbf{$v$-polynomials} (variable polynomials)\\
$$F_v=\{t_{X_i}+t_{\overline{X}_i}\ |\ i = 1, \ldots ,n \}.$$ i.e $F_v$ is a set of $v$-polynomials.

\item With respect to each clause $C_l$ in $\mathcal{I}$ for $l \in \{1, \ldots ,m\}$, we associate a polynomial. Without loss of generality assume that $C_l=(X_i \vee \overline{X}_j \vee X_k)$, for $i,j,k \in \{1, \ldots ,n\}$. The polynomial associated with $C_l$ is $$\frac{t_{X_i}x_{c_l}}{c_l}+\frac{t_{\overline{X}_j}x_{c_l}}{c_l}+\frac{t_{X_k}x_{c_l}}{c_l}.$$\\
We will refer to the above set of polynomials as \textbf{$c$-polynomials} (clause polynomials).\\
$$F_c=\Bigg \{ \frac{t_{X_i}x_{c_l}}{c_l}+\frac{t_{X_j}x_{c_l}}{c_l}+\frac{t_{X_k}x_{c_l}}{c_l} \Bigg|\ C_l=(X_i \vee X_j \vee X_k) \mbox{ is a clause in $\mathcal{I}$}\Bigg \}$$  

\item The third set of polynomials are those that contain just one term in their support:\\
$$F_1=\{t |\ \mathrm{deg}(t) = 8\},\ F_2=\displaystyle\bigcup_{i=1}^{n} (R_i \backslash K_i),\ \mathcal{F}'=F_1 \cup F_2. $$
We refer to the set of polynomials in $\mathcal{F}'$ as \textbf{$t$-polynomials} (polynomials containing just one term).
\end{itemize}
From the above set of polynomials, we construct the system of polynomials $\mathcal{F}$ which is an instance to the BBD problem:
$$\mathcal{F}=F_v \cup F_c \cup \mathcal{F}'.$$
Note that all the terms in $\mathrm{Supp}(\mathcal{F})$ have total degree either 7 or 8. Also, for any two polynomials $f,g \in \mathcal{F}$ we have $\mathrm{Supp}(f) \cap \mathrm{Supp}(g)= \phi$.
\par The construction of each polynomial in $F_c,F_v$ can be done in time polynomial in $n,m$. So $F_c,F_v$ can be constructed in time polynomial in $n$ and $m$ since $|F_c|=m$ and $|F_v|=n$. $F_1,F_2$ can be computed in time polynomial in $|F_1|$ and $|F_2|$. Also $|F_2|$ is bounded above by $\sum_{i=1}^{n}|R_i|\ (\leq \sum_{i=1}^{n}4N \leq 4nN)$ and $|F_1| \leq \left( \begin{array}{c} 8+N \\ 8 \end{array} \right) < N^8$. Hence $F_1,F_2$ can be constructed in time polynomial in $N$. Since $F_c,F_v,F_1$ and $F_2$ can be constructed in time polynomial in $N$, the reduction can be performed in polynomial time. \\
 
\noindent We state a theorem that will be helpful for proving the correctness of reduction in the next section.

\begin{theorem}
\label{varclause}
Let $\mathcal{F}$ be a $\mathcal{O}$-border basis. If $X_i$ appears in $C_l$ for $i \in \{1, \ldots ,n\},\ l \in \{1, \ldots ,m\}$ then both $t_{X_i}$ and $\frac{t_{X_i}x_{c_l}}{c_l}$ cannot be in $\boi$. Similarly if $\overline{X}_i$ appears in $C_l$ for $i \in \{1, \ldots ,n\},\ l \in \{1, \ldots ,m\}$, then both $t_{\overline{X}_i}$ and $\frac{t_{\overline{X}_i}x_{c_l}}{c_l}$ cannot be in $\boi$.  
\end{theorem}
\begin{proof}
Assume that $X_i$ appears in $C_l$. We have\\
$$k(t_{X_i}x_{c_l}) \cap \mathrm{Supp}(F_c \cup F_v)=\Big \{ t_{X_i},\frac{t_{X_i}x_{c_l}}{c_l}\Big \} \mbox{ and}$$ 
$$\ k(t_{X_i}x_{c_l}) \backslash \Big \{ t_{X_i},\frac{t_{X_i}x_{c_l}}{c_l}\Big \} \subset F_2 \enspace .$$
Since $F_2$ contains $t$-polynomials, every term in the support of $F_2$ has to be in $\boi$ and similarly all the terms in $F_1$ has to be in $\boi$. Hence,
$$t_{X_i}x_{c_l} \in \boi,\ k(t_{X_i}x_{c_l}) \backslash \Big \{ t_{X_i},\frac{t_{X_i}x_{c_l}}{c_l}\Big \} \subset \boi \enspace .$$
Now, both $t_{X_i},\frac{t_{X_i}x_{c_l}}{c_l}$ cannot be in $\boi$ without contradicting the lemma \ref{notallinbo}. Similarly, it can be argued that if $\overline{X}_i$ appears in $C_l$ then both $t_{\overline{X}_i}$ and $\frac{t_{\overline{X}_i}x_{c_l}}{c_l}$ cannot be in $\boi$.
\end{proof}

\subsection{Correctness of reduction} 
We now state the main theorem of this paper.
\begin{theorem}
3,4-SAT instance $\mathcal{I}$ is satisfiable if and only if $\mathcal{F}$ is a $\oi$-border basis with respect to some order ideal $\oi$.
\end{theorem}
\begin{proof}
Suppose $\mathcal{F}$ is an $\oi$-border basis of $\ideal$ with respect to order ideal $\oi$, we will construct an assignment to $\mathcal{I}$ and show that it is a satisfying assignment.\\
The truth values to variables in instance $\mathcal{I}$ are assigned as follows. Consider the polynomial $t_{X_i}+t_{\overline{X}_i} \in F_v$ for $i \in \{1, \ldots ,n\}$. Exactly one among the terms $t_{X_i},t_{\overline{X}_i}$ has to be in $\oi$ and the other term in $\boi$. If $t_{X_i}$ is in $\oi$, then assign true value to variable $X_i$ and if $t_{\overline{X}_i}$ is in $\oi$, then assign false value to $X_i$.\\

\noindent \textbf{Claim. }The above assignment is a satisfiable assignment to $\mathcal{I}$.\\ 
\textit{Proof.}
Assume that the above assignment is not a satisfiable assignment then there exists a clause $C_l$ for $l \in \{1, \ldots ,m\}$ such that $C_l$ is not satisfied. Without loss of generality let $C_l$ be of the form $(X_i \vee \overline{X}_j \vee X_k)$, where $i,j,k \in \{1,, \ldots ,n\}$. Since $C_l$ is not satisfied, all of $t_{X_i},t_{\overline{X}_j},t_{X_k}$ are in $\boi$. From Corollary \ref{varclause}, this implies that $\frac{t_{X_i}x_{c_l}}{c_l},\frac{t_{\overline{X}_j}x_{c_l}}{c_l}, \frac{t_{X_k}x_{c_l}}{c_l} \in \oi$. Consider the polynomial\\
$$f=\frac{t_{X_i}x_{c_l}}{c_l}+\frac{t_{\overline{X}_j}x_{c_l}}{c_l}+\frac{t_{X_k}x_{c_l}}{c_l} \in F_c \enspace .$$
All the terms in the support of $f$ are in $\oi$. But this is not possible since $\mathcal{F}$ is a border basis and $f$ should contain exactly one term in $\boi$, a contradiction.\\   

\noindent Suppose that $\mathcal{I}$ is satisfiable. Let $A$ be a satisfying assignment to instance $\mathcal{I}$. Using $A$, we will construct an order ideal $\oi$ such that $\mathcal{F}$ is a $\oi$-border basis. For that we first construct sets $\oi$ and $T$ and prove the following statements.\\
\textit{i)} $\oi$ is an order ideal,\\
\textit{ii)} $T$ is the border of the order ideal $\oi$ \textit{i.e.} $T=\boi$,\\
\textit{iii)} $\mathcal{F}$ is a $\oi$-border prebasis and\\
\textit{iv)} $\mathcal{F}$ is a $\oi$-border basis.\\ 

\noindent We construct the set $T$ as follows.\\ 
1) For $i \in \{1, \ldots ,n\}$, if $X_i$ is assigned to be false in assignment $A$ then include $t_{X_i}$ in T. If $X_i$ is assigned to be true then include $t_{\overline{X}_i}$ in $T$\\
2) Let $C_l$ be a clause in instance $\mathcal{I}$ for $l \in \{1, \ldots ,m\}$. Assume that $C_l=(X_i \vee \overline{X}_j \vee X_k)$ for $i,j,k \in \{1, \ldots ,n\}$. Associated to this clause, we have the polynomial $$f=\frac{t_{X_i}x_{c_l}}{c_l}+\frac{t_{\overline{X}_j}x_{c_l}}{c_l}+\frac{t_{X_k}x_{c_l}}{c_l} \in \mathcal{F} \enspace .$$
If one term among $t_{X_i}, t_{\overline{X}_j}, t_{X_k}$, say $t_{X_i}$, is not in $T$ (if there are more than one term among $t_{X_i}, t_{\overline{X}_j}, t_{X_k}$ not in $T$ then pick one term arbitrarily) then include $\frac{t_{X_i}x_{c_l}}{c_l}$ in $T$. Thus, in the support of every clause polynomial no more than one term is included in $T$. \\
3) Include all the terms in the support of $F_1 \cup F_2$ to be in $T$. \\

\noindent \textbf{Claim. }Let $\oi=\mathbb{T}_{\leq 8} \backslash T$. $\oi$ is an order ideal.\\
\textit{Proof.}
All the terms of total degree 8 are in $T$ (by construction). Thus, $\oi$ contains terms of total degree 7 or less. If $t \in \oi$ and $t'|t$ then $\mathrm{deg}(t') < \mathrm{deg}(t) \leq 7$ which implies that $deg(t') < 7$. But since $T \subset \mathrm{Supp}(\mathcal{F})$ and $\mathrm{Supp}(\mathcal{F})$ contains no term of total degree less than 7, all the terms of total degree 6 or less are in $\oi$. Therefore, $t' \in \oi$.\\ 

\noindent \textbf{Claim. }$T$ is the border of the order ideal $\oi$ \textit{i.e.} $T=\boi$.\\
\textit{Proof.}
Let $t' \in \boi$. There exists a term $t \in \oi$ and an indeterminate $y$ such that $t'=ty$. Since all the terms in $\oi$ have total degree 7 or less, we have $\mathrm{deg}(t) \leq 7$ which implies that $t'=ty \in \mathbb{T}_{\leq 8}$. By our construction of $\oi$, this means that $t' \in T$. This proves that $\boi \subset T$.
\par In order to show $T \subset \boi$, it is enough to show that for a term $t \in T$, there exists an indeterminate $y$ such that $y|t$ and $\frac{t}{y}=t' \notin T$ \textit{i.e.} $t' \in \oi$. Now, since all the terms of total degree 6 or less are in $\oi$, all the terms of total degree 7 in $T$ are also in $\boi$. So, assume that there exists a term $t$ such that $\mathrm{deg}(t)=8$ and $k(t) \subset T$. We prove by contradiction that such a term cannot exist. Since all the terms of total degree 7 in $T$ are in $\cup_{i=1}^{n}R_{i}$, $k(t) \subset \cup_{i=1}^{n}R_{i}$. From Lemma \ref{regsamep}, $k(t)$ should be a subset of $R_{i}$ for some $i \in \{1, \ldots ,n\}$. There are two cases for $t$ as described below.\\
\textit{(i)} $t \in P_i$: Assume without loss of generality, $t=t_{X_i}x_{c_l} \in P_{X_i}$ for $i \in \{1, \ldots ,n\}, l \in \{1, \ldots ,m\}$. By our construction, both $t_{X_i}$ and $\frac{t_{X_i}x_{c_l}}{c_l}$ cannot be in $T$. Hence atleast one child of $t$ is in $\oi$ and thus not all terms in $k(t)$ is contained in $T$. So, this case is not possible.\\ 
\textit{(ii)} $t \notin P_i$: From Corollary \ref{boundsamep}, we have
\begin{eqnarray*}
|k(t) \cap R_i| & = & |k(t)| \leq |(P_{X_i} \cup P_{X_i})| \\
& \Rightarrow & |k(t)| \leq |P_i|\\
& \Rightarrow & |I(t)| \leq |P_i|.
\end{eqnarray*}
Now, for any term $t' \in k(t)$ we have $I(t') \leq |P_i|$. But from Lemma \ref{boundind}, $I(t'') \geq |P_i|+2$ for any term $t'' \in k(P_i)=R_i$. Thus this case is not possible. \\    
From the above two cases we get a contradiction that there exists a term $t$ such that $k(t) \subset R_i$ for some $i \in \{1, \ldots ,n\}$ and thus $k(t) \nsubseteq T$. So, $t$ has atleast one child in $\oi$. Thus, $T \subset \boi$.\\

\noindent \textbf{Claim. }$\mathcal{F}$ is a $\oi$-border prebasis.\\
\textit{Proof.}
In order to show $\mathcal{F}$ is a $\oi$-border prebasis, we have to show that each polynomial in $\mathcal{F}$ has exactly one term in $\boi$ and the rest of the terms in $\oi$. We show this for all the polynomials in $\mathcal{F}$:
\begin{itemize}
\item $t$-polynomials: From our construction, all the terms in the t-polynomials are in $T$ \textit{i.e.} in $\boi$ and hence each polynomial has exactly one term in $\boi$.

\item $v$-polynomials: Again by our construction, each $v$-polynomial has exactly one term in $T$ \textit{i.e.} $\boi$ and the other term in $\oi$.

\item $c$-polynomials: Consider a clause $C_l$ for $l \in \{1, \ldots ,m\}$. Assume that $C_l=(X_i \vee \overline{X}_j \vee X_k)$ where $i,j,k \in \{1, \ldots ,n\}$ in the instance $\mathcal{I}$. Let $f$ be the polynomial associated with the clause $C_l$:
$$f=\frac{t_{X_i}x_{c_l}}{c_l}+\frac{t_{\overline{X}_j}x_{c_l}}{c_l}+\frac{t_{X_k}x_{c_l}}{c_l} \in \mathcal{F}.$$
Since all the terms in the support of $f$ have total degree 7, the terms must either be in $\boi$ or $\oi$. Consider the following cases:\\
\textit{Case (i):} More than one term in $f$ is in $\boi$: this cannot happen from our construction.\\
\textit{Case (ii):} All the terms are in $\oi$: This can happen only if all of $t_{X_i},t_{\overline{X}_j},t_{X_k}$ are in $\boi$ which implies that $X_i,\overline{X}_j,X_k$ are false in assignment $A$. So, $C_l$ is false. But this is not possible since assignment $A$ satisfies instance $\mathcal{I}$. Hence this case is not possible.\\
From the above two cases, we deduce that exactly one term in the support of $f$ belongs to $\boi$ and from our construction, rest of the terms in $f$ must belong to $\oi$.
\end{itemize}
Since any polynomial in $\mathcal{F}$ must be either a $t$-polynomial, $c$-polynomial or $v$-polynomial, from the above argument we deduce that $\mathcal{F}$ is a $\oi$-border prebasis.\\

\noindent \textbf{Claim. }$\mathcal{F}$ is a $\oi$-border basis of $\ideal$.\\
\noindent \textit{Proof.}
Since $\mathcal{F}$ is a $\oi$-border prebasis, if $\mathcal{F}$ satisfies Buchberger criterion for border basis then $\mathcal{F}$ is a $\oi$-border basis. Thus we need to show that for any two neighbouring polynomials $f,g \in \mathcal{F}$, $S(f,g)$ can be written as a linear combination of polynomials in $\mathcal{F}$. Before we consider the following cases for $f$ and $g$ we note that any polynomial containing only terms of total degree 8 in it's support can be expressed as a sum of $t$-polynomials in $F_1 \subset \mathcal{F}$. Thus, in order to prove that $\mathcal{F}$ satisfies Buchberger criterion it is enough to show that the support of $S(f,g)$ contains only terms of total degree 8. Neighbouring polynomials $f,g$ can be of the following cases,\\ 
\textit{Case (i):} $f$ and $g$ are $t$-polynomials: then $S(f,g)=0$.\\
\textit{Case (ii):} $f$ is a $t$-polynomial and $g$ is a $c$-polynomial or a $v$-polynomial: All the terms in $\mathrm{Supp}(g)$ have total degree 7. Hence for any indeterminate $y$, all the terms in $\mathrm{Supp}(yg)$ are of total degree 8. If $f \in F_2$, then $yf$ for any indeterminate $y$ is also a $t$-polynomial of total degree 8. The S-polynomial of $f$ and $g$ can be\\
$$S(f,g)=f-y_1g$$ or
$$S(f,g)=y_2f-y_1g,$$
for some indeterminates $y_1,y_2$. In the first case, $f$ has to be in $F_1$ (if $f$ were to be in $F_2$, by the way we have written the S-polynomial the border term of total degree 7 in $f$ is equal to $y_1b$ of total degree 8 where $b$ is the border term in $g$ which is not possible) and hence support of $S(f,g)$ contains only terms of total degree 8. The second case can happen only if $f \in F_2$ and hence support of $S(f,g)$ contains only terms of total degree 8.\\
\textit{Case (iii):} $f$ and $g$ are not $t$-polynomials: S-polynomial of $f$ and $g$ is of the form,
$$S(f,g)=y_1f-y_2g,$$
for some indeterminates $y_1,y_2$. As argued before, all the terms in the support of $y_1f$ and $y_2g$ are of total degree 8. Hence, all the terms in the support of $S(f,g)$ contains only terms of total degree 8. From the three cases it follows that $\mathcal{F}$ is a $\oi$-border basis of $\ideal$.\\
\end{proof}
Thus, we have proved that $I$ has a satisfying assignment if and only if $\mathcal{F}$ is a $\oi$-border basis of $\ideal=\langle \mathcal{F} \rangle$ for some order ideal $\oi$. There is a polynomial time reduction from 3,4-SAT instance to BBD instance and since 3,4-SAT is NP-complete, we have the result that BBD is NP-complete.

\section{Conclusion}
In this paper we introduced the Border Basis Detection and proved it to be NP-complete.
 
{\footnotesize
\section*{\footnotesize Acknowledgments}
Authors would like to thank Vikram M. Tankasali and Prashanth Puranik for their useful comments and suggestions on this work.
}

\bibliographystyle{plain}
\bibliography{prabhanjan}

\end{document}